\documentclass[runningheads]{llncs}

\usepackage[utf8]{inputenc}
\usepackage{amssymb}
\usepackage{amsmath}
\usepackage{graphicx}
\usepackage{xcolor}
\usepackage{booktabs} 
\usepackage{longtable}
\usepackage{tabularx}
\usepackage{subcaption}
\usepackage[ruled,vlined,linesnumbered]{algorithm2e}
\usepackage{hyperref}
\DeclareMathOperator*\bigcircop{\bigcirc}

\newcommand{\acronym}{ERDALT}

\begin{document}

\title{Certifiably robust malware detectors by design}
\author{Pierre-François Gimenez\inst{1}, Sarath Sivaprasad\inst{2}, Mario Fritz\inst{2}}


\institute{Univ. Rennes, Inria, IRISA, CentraleSupélec, Rennes, France\\
\email{pierre-francois.gimenez@inria.fr}\\
 \and
CISPA Helmholtz Center for Information Security, Saarbrücken, Germany\\
\email{\{sarath.sivaprasad,fritz\}@cispa.de}}

\maketitle

\begin{abstract}
Malware analysis involves analyzing suspicious software to detect malicious payloads. Static malware analysis, which does not require software execution, relies increasingly on machine learning techniques to achieve scalability. Although such techniques obtain very high detection accuracy, they can be easily evaded with adversarial examples where a few modifications of the sample can dupe the detector without modifying the behavior of the software. Unlike other domains, such as computer vision, creating an adversarial example of malware without altering its functionality requires specific transformations. We propose a new model architecture for certifiably robust malware detection by design. In addition, we show that every robust detector can be decomposed into a specific structure, which can be applied to learn empirically robust malware detectors, even on fragile features. Our framework \acronym{} is based on this structure. We compare and validate these approaches with machine-learning-based malware detection methods, allowing for robust detection with limited reduction of detection performance.

\end{abstract}

\section{Introduction}

Malware infection is a major cybersecurity threat that evolves quickly. According to the French National Agency for the Security of Information Systems (ANSSI), there were 144 successful ransomware attacks in France alone in 2024.
This threat is fought with malware analysis techniques developed by industry and academia.


These techniques can be broadly categorized into two types: static analysis and dynamic analysis. 
On the one hand, dynamic analysis executes the code in a controlled sandbox. This technique requires more equipment and time, but the malware cannot easily hide its malicious payload when activated. However, malware can evade dynamic analysis by verifying whether it is running in a virtual machine or any monitoring environment. 
On the other hand, static analysis relies on descriptive features such as section description, control-flow graph, and opcodes n-grams without running the program.
This analysis is widely used because it is fast and inexpensive. 
For these reasons, our article focuses on the static analysis of malware targeting Windows, the most popular and targeted desktop operating system.

Many traditional detection techniques used by antivirus are based on pattern matching, relying on indicators of compromise from threat intelligence (file hashes, IP addresses of command-and-control servers, suspicious domains, etc.). 
However, these techniques are difficult to scale up and cannot reliably detect new variants or malware families. Machine learning techniques have been very useful in addressing these two issues: these methods can achieve very high accuracy~\cite{raff2018malware}.
However, they are vulnerable to attacks called \textit{adversarial examples}. 
Adversarial examples are samples with intentional perturbations optimized to make the model give an incorrect prediction. 
While first identified in the field of computer vision \cite{szegedy2013intriguing} 
such adversarial examples have been found in many domains, including malware analysis \cite{suciu2019exploring}.

Several methods have been proposed to improve the robustness of machine learning models against adversarial attacks, such as \cite{advtraining,cohen_certified}, 
but they do not take advantage of domain specificity. The key constraint in designing adversarial samples for evading malware detection comes from the discrete nature of malware and the highly constrained set of transformations the attacker can use to evade malware detectors. Besides, most of the methods with guaranteed robustness assume that the perturbation is small. 
They are not applicable here, since malware attackers only need the perturbation to conserve functionality, and the magnitude of perturbation itself is insignificant.

Based on this assessment, we propose a framework where the attacker's capability is not limited by the magnitude of perturbation but by the limited set of transformations they can apply to the binary. We propose a robust-by-design malware detector that relies on features that an attacker cannot arbitrarily modify, no matter how many transformations they use. We extend this proposition by characterizing certifiably robust malware detectors as a combination of a preprocessing function and a monotonically increasing function, from which we derive a new model architecture, leading to a new defense against adversarial examples: the framework \acronym{} ("Empirically Robust by Design with Adversarial Linear Transformation"). 
This framework extracts knowledge about the attacker's capability from examples of adversarial attacks, leading to an empirically robust-by-design malware detector. Experiments show the trade-off between robustness and detection performances for various models with several defenses.

The contributions of this paper can be summarized as follows:
\begin{itemize}
\item a characterization of certifiably robust detectors;
\item \acronym, a framework to learn empirically robust detectors.
\end{itemize}

The paper is structured as follows. Section \ref{sec:robust-by-design} presents our analysis of certifiable robust malware detectors by design. The application to empirically robust malware with the \acronym~framework is presented in Section \ref{sec:empirically-robust}. Section \ref{sec:experiments} details our experiments. Section \ref{sec:conclusion} concludes the article.

\section{Related work}
\label{sec:related}

While most commercial anti-viruses rely on a database of known signatures to recognize malware, the surge of machine learning in recent decades allows detectors to identify unseen malware that would not correspond to any signature. 
Feature extractions for malware analysis are generally split into static and dynamic categories. We focus on static analysis. 
Static features are extracted from the binary without executing it. It is fast and harmless, but obfuscating techniques (such as packing) can hide part of the content of the malware. The static features proposed in EMBER \cite{anderson2018ember} are widely used in machine learning models. 

Deep learning models have been adopted in many applications, including  
malware analysis~\cite{vinayakumar2019robust}. Evasion attacks have also gained attention among both researchers and practitioners \cite{advtraining}. 
Two types of techniques have been proposed to craft adversarial examples: black-box attacks that have only access to the output of a detector and white-box attacks that know the architecture and weights of the target model. 
Various research works have successfully performed adversarial attacks to evade malware detectors and show how vulnerable some classifiers are~\cite{ling2021adversarial}. Two families of evasion attacks have been encountered: attacks on problem space and attacks on feature space \cite{pierazzi2020intriguing}. Attacks on problem space seek to directly create new binaries that evade classifiers, while attacks on feature space seek to create features that evade classifiers. This distinction is irrelevant, e.g., for computer vision, where the feature extraction is reversible, so it is trivial to transpose an attack on an image file to an attack on image features, and vice versa. However, static and dynamic feature extractions in malware analysis are generally neither reversible nor differentiable. So, even if the attacker crafts an adversarial features vector, it is not trivial to construct a malware sample with the corresponding features. Attack on problem space (i.e., on binaries) is almost always black-box attacks due to the non-invertible and non-differentiable feature mappings that make classical white-box evasion attacks like FGSM impossible. 

General defense methods against adversarial attacks have been proposed \cite{silva2020opportunities}, such as adversarial training, regularization approach, gradient masking, and adversarial example detection.  For Android malware analysis, \cite{demontis2017yes} proposes a Lipschitz-bounded linear classifier to improve the robustness. For Windows malware analysis, \cite{huang2023certified} uses randomized smoothing to improve the robustness, but it significantly lowers the detector's accuracy. Similarly, \cite{bena2024certifying} proposes a certification scheme for dynamic malware detectors.
However, these methods are limited to perturbations with a small magnitude, which is not realistic for malware adversarial examples. 
Some work on malware analysis avoids assuming the perturbation has a small magnitude, such as \cite{incer2018adversarially} who relies on monotonic models to be robust against some attacks, at the cost of a significant detection performance loss.
Our work expands on their idea of using the monotonicity property to provide robustness against adversarial examples. 


\section{Certifiable robust detector by design}
\label{sec:robust-by-design}

In domains like computer vision, the typical assumption in adversarial attacks is that the perturbation is small, i.e., within an $\epsilon$-ball \cite{akhtar2021advances}. This assumption is depicted in Figure \ref{fig:perturbation} (left part), where any picture $P'$ within a $\epsilon$-ball centered around $P$ is an adversarial example candidate for $P$.
For example, a spam that poses as an official bank message could include a bank logo to appear legitimate. Spam detectors detect such impersonation by verifying the consistency between logos and source email domains. An attacker could craft an adversarial example so that the spam detector does not recognize the bank logo properly. Since the magnitude of the perturbation is small, the logo would still look genuine to the user, and the impersonation would be successful. However, this hypothesis does not hold in malware analysis. Because the attacker’s target won’t "look" at the binary content of the malware, the utility of the modified malware does not decrease significantly with the size of the perturbation.


\begin{figure*}
    \centering
    \includegraphics[width=\textwidth]{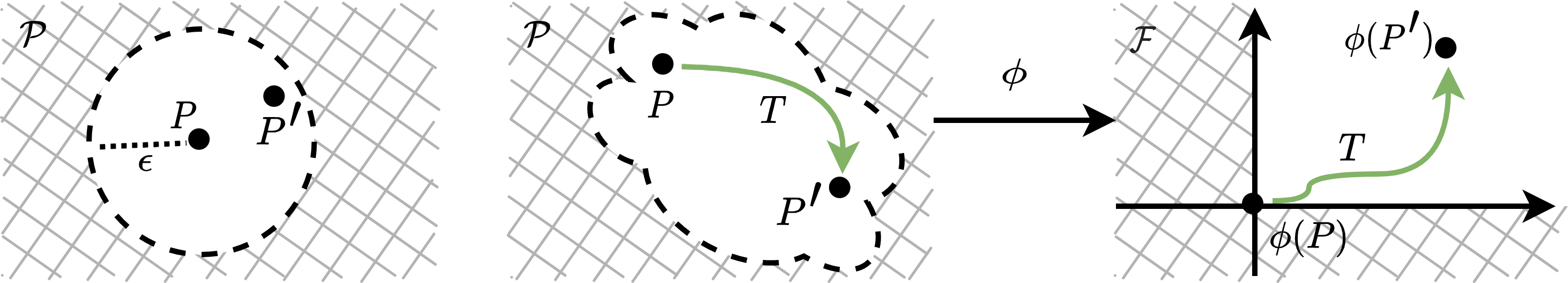}
    \caption{\parbox[t]{0.25\textwidth}{(a) Adversarial perturbation in the image problem space. } \hfill \parbox[t]{0.6\textwidth}{(b) Adversarial perturbation in the problem and the feature spaces. The crosshatched part of the space cannot be reached by perturbation.}}
    \label{fig:perturbation}
\end{figure*}



However, attackers still require the modified malware to function properly (i.e., stealing credentials, encrypting disks, etc.). So, adversarial attacks against malware analysis are generally performed in the problem space (i.e., the space of executable binaries) instead of the feature space, as it is the case for computer vision. Such attacks generally rely on elementary binary transformations 
to alter the malware in the problem space. This is depicted in the middle part of Figure \ref{fig:perturbation}, where the set of adversarial example candidates (in white) is the set of binaries with the same behavior. In this example, a program $P$ is transformed into a functionally equivalent program $P'$ with the transformation $T$.
Plenty of modifications are available to the attacker, and they do not all have the same cost. For example, off-the-shelf tools can automatically add sections to a binary, making such adversarial transformations easy to automate \cite{demetrio2021functionality}, but modifications like static import removal require access to the malware source code, which may not be the case in the malware-as-a-service business. However, almost any feature collected in malware static analysis can be modified by bloating the malware, such as adding useless sections or system call.

From this analysis, we propose a new set of features manually selected to be difficult to change for an attacker with no access to the malware source code. They contain indicators of adversarial transformation, such as DOS stub modification and signature removal, as well as features that are difficult to reduce automatically, like the number of sections, the total entropy of a section, the number of statically imported functions, and the number of keywords used in imported functions. 
There are only 40 features, which is very limited compared to classical feature mapping like EMBER, which contains several thousands of features. In the following, we will use this set of features as a baseline for how robust and accurate classifiers are when learned with a feature mapping designed with adversarial examples in mind. 

We rely on the following formal definitions of adversarial attacks and robust detectors to prove that monotonic models can lead to robust-by-design detectors. We also prove the converse is true: any robust detector can be decomposed as a monotonic detector on a latent feature space. 
We introduce $M$, the transformations available in a threat model. 
From this set, we define the preorder $\preceq_M$ in the space of programs such that $P \preceq_M P'$ if the program $P$ can be transformed into $P'$ with transformations available in $M$. We can remark that this preorder is reflexive and transitive, but it is generally not complete (we may have $P_1 \npreceq_M P_2$ and $P_2 \npreceq_M P_1$) nor anti-symmetric (we may have $P_1 \preceq_M P_2$ and $P_2 \preceq_M P_1$). 
We can formally define a robust detector for such a preorder as a detector whose decision cannot be modified from "malicious" to "benign" by applying transformations from $M$ on the original program.


As mentioned by \cite{pierazzi2020intriguing}, the notions of problem space and feature space are prevalent in malware analysis. 
By far, the most common approach in machine learning for malware analysis is to extract features from programs. Let $\phi$ be a mapping from the problem space $\mathcal P$ to the feature space $\mathcal F_\phi$. 
We propose the following definition:

\begin{definition}\label{def:robust}
Let $M$ be a set of transformations, $\phi: \mathcal P \rightarrow \mathcal F$ a feature mapping, $f: \mathcal F \rightarrow \mathbb R$ a classifier and $\tau$ its decision threshold. $f$ is said to be robust against adversarial attacks if, for any program $P$ and $P'$ such that $P \preceq_M P'$, $f(\phi(P)) \geq \tau \Longrightarrow f(\phi(P')) \geq \tau$.
\end{definition}

This definition allows us to reason separately about the feature mappings $\phi$, generally standardized within the domain, and the classifiers $f$, which can be based on many machine learning techniques. Indeed, the classifiers are typically general-purpose algorithms and models that cannot exploit the specific structure of the problem space of binaries. This is the role of the feature mapping, written by malware experts, to extract relevant data. For this reason, we argue that the adversarial examples issue should also be solved in the feature mapping itself.

Following this guideline, we identify two ways for obtaining robust classifiers: 1) only extract features that are hard for the attacker to modify and use any classifier, and 2) only extract monotonic features, i.e., that the attacker can only increase, and use a monotonic classifier. As discussed previously, almost every feature is easy for the attacker to modify, at least partially, so the first strategy is moot for malware analysis. The second strategy is based on a monotonic feature mapping that ensures that, if $P \preceq_M P'$, then $\phi(P) \leq \phi(P')$.
%
This second strategy is illustrated in Fig. \ref{fig:perturbation}: with a monotonic feature mapping $\phi$, any transformation $T$ in the problem space is mapped to a transformation that cannot decrease the value of any feature. So, given some binary $P$, the attacker can only produce an adversarial example in the first quadrant relative to $P$. It is straightforward that the monotonicity of $\phi$ and $f$ are indeed sufficient for $f$ to be robust.

\begin{proposition}\label{prop:mono-is-robust}
Let $M$ be a threat model, $\phi$ a feature mapping and $f$ a classifier. If $\phi$ is monotonically increasing w.r.t. $\preceq_M$ and if $f$ is monotonically increasing w.r.t. $\leq$, then $f$ is robust against adversarial attacks for the threat model $M$.
\end{proposition}

The feature mapping we propose, discussed earlier, is monotonic when the attacker has no particular capability. Therefore, robustness is guaranteed when such a feature set is used with a monotonic classifier (this claim is experimentally assessed in Section~\ref{sec:experiments}). Remark that some common feature engineering strategies, such as histograms, are typically not monotonic. For example, the EMBER features set \cite{anderson2018ember} contains bytes histograms of readable strings. However, even though it is difficult for the attacker to remove readable strings, they can easily reduce the proportion of one byte by adding new strings with other bytes. For this reason, such features are fragile and should be avoided.

In the following, we show that such monotonic classifiers are not just yet another tool to obtain robustness against adversarial examples. In fact, with the proper feature post-processing, every robust classifier can be interpreted as a monotonic classifier, as shown by the following proposition:

\begin{proposition}\label{th2}
Let $M$ be a threat model, $\phi$ a feature mapping and $f$ a classifier such that $f$ is robust against adversarial attacks for the threat model $M$. There exist $g$ and $h$ such that $f \circ \phi = (f \circ h) \circ (g \circ \phi)$ and $f \circ h$ is monotonically increasing.
\end{proposition}

\begin{proof}
Let $\mathcal P$ the set of executable binaries and $\mathcal D_\phi$ the codomain of $\phi$.
Consider the following complete preorder $\preceq_\phi$ defined on $\mathcal D_\phi \times \mathcal D_\phi$: $\phi(P) \preceq_\phi \phi(P')$ if and only if $f(\phi(P)) \leq f(\phi(P'))$. Since $f \circ \phi$ is robust by assumption, this preorder satisfies the property: $P \preceq_M P' \Rightarrow \phi(P) \preceq_\phi \phi(P')$.
%
Denote $k$ the dimension of the feature space and $g$ any strictly monotonically increasing function from $(\mathcal D_\phi,\preceq_\phi)$ into $(\mathcal D_f,\leq)$, where $\mathcal D_f = g(\mathcal D_\phi) \subseteq \mathbb R^k$. 
So, if $v_1 \prec_\phi v_2$, then $g(v_1) < g(v_2)$. Remark that $g$ is surjective (by definition of its codomain) but generally not injective: indeed, two programs $P_1$ and $P_2$ such that $\phi(P_1) \preceq_\phi \phi(P_2)$ and $\phi(P_2) \preceq_\phi \phi(P_1)$ will be mapped to the same vector due to the anti-symmetry of $\leq$. Let $h$ be a function defined on $\mathcal D_f$ such that $h(v) \in g^{-1}(v)$ (the latter is a set because $g$ is generally not injective). Let us show that $f \circ \phi = f \circ h \circ g \circ \phi$. Let $P \in \mathcal P$, let $v_1 = \phi(P)$ and $v_2 = h(g(v_1))$. Due to the definition of $g$ and $h$, $g(v_1) = g(v_2)$. Since $\preceq_\phi$ is a complete preorder, there are only three possibilities: $v_1 \prec_\phi v_2$ (but in that case, $g(v_1) < g(v_2)$), $v_2 \prec_\phi v_1$ (but in that case, $g(v_2) < g(v_1)$), or $v_1 \sim_\phi v_2$. Only the last case is possible. Due to the definition of $\preceq_\phi$, we can conclude that $f(v_1) = f(v_2)$, i.e., $f(\phi(P)) = f(g(h(\phi(P)))$. Therefore, $f \circ \phi = f \circ h \circ g \circ \phi$.
Let us now show that $f \circ h$ is monotonically increasing. Let $v_1, v_2 \in \mathcal D_\phi$ such that $v_1 \leq v_2$. By definition of $h$ and $g$, $h(v_1) \preceq_\phi h(v_2)$. Due to the definition of $\preceq_\phi$, it implies that $f(h(v_1)) \leq f(h(v_2))$. Therefore, $f \circ h$ is monotonically increasing, concluding the proof.
\end{proof}


So, if there exists a robust classifier $f$ with good performances, then with the correct feature post-processing $g$, one can learn a monotonic detector on the features $g \circ \phi$ and obtain the same performances as $f$ while keeping the certifiable robustness. Such monotonic classifiers have been proposed for malware analysis \cite{incer2018adversarially}. Our work shows that 1) this idea stems from the asymmetrical cost of the transformations in the problem space, 2) such classifiers can be certified to be robust with respect to a threat model, and 3) all robust classifier can be expressed as a monotonic classifier with some change to the feature mapping. The next section shows how we leverage Proposition \ref{th2} to propose a new framework, \acronym, that learns feature post-processing alongside a monotonic classifier to obtain a more robust classifier.

\section{\acronym: an empirically robust by-design malware detector}
\label{sec:empirically-robust}

The method described in Section \ref{sec:robust-by-design} is certifiably robust, assuming the risk analysis is correctly updated as attacker capabilities evolve. In this Section, we propose a new detector that is \textit{not} certifiably robust but only empirically robust. Indeed, it requires adversarial examples to learn the set of transformations the attackers can use and, therefore, is limited by its training dataset. 

Since any robust classifier can be decomposed into a post-processing function $g$ and a monotonic function $f \circ h$, we propose to learn these two functions $g$ and $f \circ h$ jointly to obtain a robust classifier. 
To perform this learning, we assume we have access to adversarial examples alongside original samples so the model can learn what the attacker can (and cannot) do. This assumption is akin to what requires other protection techniques, such as adversarial training.

To illustrate this idea, suppose the attacker can replace API calls with equivalent ones, such as replacing \texttt{CreateFile} with \texttt{CreateFileEx}. In that case, the two features that count the number of \texttt{CreateFileEx} and \texttt{CreateFile} calls are not monotonically increasing, but the effect of the transformation on the features is linear. Intuitively, the post-processing $g$ could build a feature that counts the occurrences of either \texttt{CreateFileEx} or \texttt{CreateFile}. This feature would not be affected by the transformation and could be used with a monotonic classifier.

In this section, we call \textit{perturbation vector} the difference of features between the transformed software and the base software, i.e., any $\phi(T(P)) - \phi(P)$ for any $T$ and $P$. For example, if a feature is unchanged by $T$, then its value is 0 in the perturbation vector. In the previous example of the transformation by substitution, the perturbation vector will be 1 for the count of \texttt{CreateFileEx}, -1 for the count of \texttt{CreateFile}, and 0 for the rest of the vector.

To simplify the learning and help the generalization, we propose to introduce an assumption that would still allow to tackle this kind of transformations by substitution: we assume perturbation vectors related to one transformation does not depend on the modified software itself. This is the case for the previous example: the perturbation vector does not depend on the original software. This hypothesis can be mathematically formalized as follows: for a threat model $M$, the feature mapping $\phi$ is such that the set of perturbation vectors, i.e. $\{\phi(T(P)) - \phi(P) \mid T \in M, P \in \mathcal P\}$, is finite. We denote this set of perturbations vectors $\Delta_M$.


\begin{figure}
\centering
\includegraphics[width=0.7\linewidth]{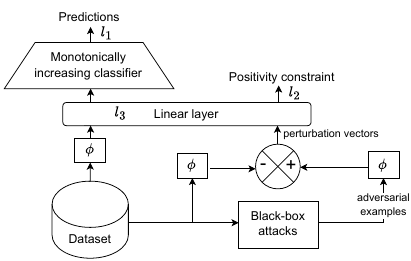}
\caption{\acronym{} framework. 
$l_1$ minimize the detection error, $l_2$ ensures the perturbations positivity and $l_3$ encourages the linear layer to be a diagonal matrix.}
\label{fig:arch}
\vspace*{-10pt}
\end{figure}

We propose the deep learning architecture presented in Figure \ref{fig:arch} to learn an empirically monotonic feature mapping from examples of attacks. It is composed of a linear layer (corresponding to the post-processing function $g$) and a monotonic classifier (corresponding to $f \circ h$). The linear layer must be a linear function, without any bias nor activation function. It is fitted such that, when the input is the perturbation vectors, its output is always positive. Such behavior is enforced with the $l_2$ loss. This constraint guarantees the monotonicity of the feature post-processing. The upper layers should be monotonic. 
Its loss $l_1$ is any typical loss function used for classification. This framework is robust by design, as shown by the following proposition. We name it “Empirically Robust by Design with Adversarial Linear Transformation” (\acronym~for short).

\begin{proposition}\label{prop:robustness}
Let $M$ be a threat model and $\phi$ a feature mapping such that $\Delta_M = \{\phi(T(P)) - \phi(P) \mid T \in M, P \in \mathcal P\}$ is finite. If $\Delta_M$ is known during training, then \acronym~is robust.
\end{proposition}
\begin{proof}
Let us decompose the architecture in two parts: the first dense layer $L$, that have no bias nor activation, and the upper layers $L’$. The first layer is constrained such that, for all perturbation vector $\delta \in \Delta_M$, $L(\delta) \geq 0$. The upper layers have a monotonicity constraint, so if $v_1 \leq v_2$ then $L’(v_1) \leq L’(v_2)$. Let us show that $L’ \circ L$ is robust. Let $P$ and $P'$ be two programs such that $P \preceq_M P'$. It means there exists a succession of $n$ transformations $T_1, T_2, \ldots, T_n$ such that $P' = (T_n \circ \ldots \circ T_2 \circ T_1) (P) = (\bigcirc_{k=n}^{1} T_k) (P)$. So:
\begin{align*}
\phi(P') - \phi(P) &= \phi((\bigcircop_{k=n}^{1} T_k) (P))\\
&= \phi((\bigcircop_{k=n}^{1} T_k) (P)) + \sum_{j=1}^{n-1} \phi((\bigcircop_{k=j}^{1} T_k)(P)) - \sum_{j=1}^{n-1} \phi((\bigcircop_{k=j}^{1} T_k)(P)) - \phi(P)\\
\end{align*}

\begin{align*}
&= \sum_{j=1}^{n} \phi((\bigcircop_{k=j}^{1} T_k) (P)) - \sum_{j=0}^{n-1}  \phi((\bigcircop_{k=j}^{1} T_k)(P))\\
&= \sum_{j=1}^{n} \left(\phi((\bigcircop_{k=j}^{1} T_k) (P)) - \phi((\bigcircop_{k=j-1}^{1} T_k)(P))\right)\\
&= \sum_{j=1}^{n} \left(\phi(T_j\circ(\bigcircop_{k=j-1}^{1} T_k) (P)) - \phi((\bigcircop_{k=j-1}^{1} T_k)(P))\right)\\
\end{align*}

We can remark that this last sum only contains elements of $\Delta_M$. Let us denote them $\delta_i$. Due to the linearity of $L$:
$L(\phi(P'))-L(\phi(P)) = L\left(\sum_i \delta_i\right) = \sum_i L(\delta_i) \geq 0$.
So, $L(\phi(P)) \leq L(\phi(P'))$. Since $L’$ is monotonically increasing, $L’(L(\phi(P))) \leq L’(L(\phi(P')))$. This proves the robustness of $L’ \circ L$.
\end{proof}

The linear layer of \acronym{} can make linear combinations of its input but will also drop the features that cannot be made monotonic with respect to the threat model. In this sense, it also serves as an automatic feature selection. 
The advantage of this method is that it does not require expert knowledge: as long as adversarial examples are available, possibly by using off-the-shelf tools or adversarial examples found in the wild, this method can be used. Its structure also makes it more explainable than adversarial training. 

%
%
%
%
%

\section{Experiments}
\label{sec:experiments}


In this section, we experimentally assess the contributions on robust classifiers, including the framework \acronym. We seek to answer the following research questions:
%
What is the impact of the features on the detection performances and the robustness performances?
What protection method allows for the best tradeoff between detection performances and robustness?
How does each component of \acronym{} contribute to its performances?
How do the PV feature selection and the linear layer of \acronym{} compare?



There are no standard PE executable datasets for malware analysis. Windows malware datasets, like EMBER2017, EMBER2018 or SOREL-20M, do not contain actual executable samples but only features. We cannot use these datasets since we compare methods with various feature sets. So, we rely on the dataset of \cite{dambra2023decoding} containing 670 families of malware, each one with 100 different samples, for a total of 67000 malware. The goodware dataset contains 16611 samples, collected from Windows default installation and various online repositories like Chocolatey. In the following, we create a balanced training dataset containing 120 families of malware (so 12000 samples) and 12000 goodware. The testing set is composed of 4611 goodware and 550 families of malware (so 55000 samples). 



We use the classical area under the ROC (or ROC AUC for short) metric to compare the detection performances. The ROC AUC is the area under the curve that plots the false positive rate against the true negative rate with respect to a decision threshold. ROC AUC typically ranges from 0.5 (random classifier) to 1 (perfect discrimination) in binary classification.
We evaluate two groups of models: 1) classical models widely used in machine learning: neural networks, $k$-nn, random forests and gradient-boosted trees, 2) monotonic models: monotonically increasing gradient-boosted trees and monotonic neural network~\cite{siva2021curious}.
The implementation is available online\footnote{\url{https://github.com/PFGimenez/certifiably-robust-malware-detectors-by-design}}.


We experiment with four protection methods against adversarial attacks: adversarial training, feature selection, our framework \acronym, and a hybrid between \acronym~and adversarial training. Remark that all these methods require adversarial examples so they can be fairly compared. Each method has access to 1033 adversarial examples.
We propose to compare \acronym~with another approach based on feature selection. Given a dataset of perturbation vectors $\Delta$, this procedure identifies all features with non-negative perturbation, i.e., that are monotonic with respect to this set of attacks. When used jointly with a monotonic model, such models should be robust, as shown by Proposition~\ref{prop:mono-is-robust}. We call this method “perturbation-vectors-based features selection”, or PV feature selection for short.



The adversarial attacks rely on the secml-malware package developed by \cite{demetrio2021functionality}. These attacks are black-box functionality-preserving transformations on Windows malware aiming at evading static analysis. They are very effective at evading detectors, even with minimal modifications. 
%

\subsection{What is the impact of the features set on robustness and performances?}

\begin{table*}[t]
\small
\begin{center}
\begin{tabular}{ l | cc | cc }

\textbf{Model} & \multicolumn{2}{c|}{\textbf{Manual features}} &	\multicolumn{2}{c}{\textbf{EMBER}}\\ 
\midrule

\textbf{}&\textbf{ROC} &\textbf{Robustness} &\textbf{ROC} &\textbf{Robustness} \\ 
\midrule

Baseline network & 89.9\% & \textbf{100\%} & 91.6\% & \textbf{82.0\%} \\
Monotonic network & 69.0\% & \textbf{100\%} & 87.4\% & 71.5\% \\
Random Forest & \textbf{94.6\%} & 98.5\% & 96.2\% & 81.0\% \\
$k$-nn & 83.7\% & 93.5\% & 88.6\% & 0\%\\
Monotonic gradient-boosted trees & 76.2\% & \textbf{100\%} & 92.7\% & 73.5\%\\
Gradient-boosted trees & 92.3\% & 99.0\% & \textbf{97.5\%} & 75.0\% \\
\end{tabular}
\end{center}
\caption{Comparing the performance (ROC AUC and robustness) of different models over the two features set without specific adversarial protection}
\label{tab:acc-rob}
\vspace*{-15pt}
\end{table*}

In this experiment, we evaluate the robustness and the detection performances of various models applied to both feature sets. For each experimental setup, we ran the eight attacks on 200 malware samples correctly identified as malware. Each attack has a 60-second timeout for scalability's sake.
The robustness is defined as the proportion of decisions evaded by at least one attack.
Table~\ref{tab:acc-rob} presents the overall robustness and detection performances. 
Most models using EMBER have a good AUC (higher than 90\%), which is typical in malware detection. However, these models are generally vulnerable to adversarial attacks. $k$-nn is the most vulnerable model (all attacks succeeded). The other models models can resist some attacks. The baseline neural network is the most robust model on EMBER by a slight margin. It has a lower ROC AUC than non-deep models. This is a common observation in malware static analysis.

Our manually selected features yield more robust models: they all have at least 93\% robustness, and some of them were never evaded. Remark the significant impact of feature mapping on the robustness: $k$-nn classifier is typically not robust (all malware have successfully evaded the detectors based on EMBER), but its robustness with manual robustness is 93.5\%. Besides, as expected, using the manual features with a monotonic model allows for 100\% robustness, although the classification performance is much lower. The best trade-off between ROC AUC and robustness is, in our opinion, the random forest model with manual features, with 94.6\% AUC and 98.5\% robustness.


This experiment shows the drastic effect of the features set on both the detection performances and the robustness of the models: even models that are easy to attack can become robust with manually selected features. However, detection performances and robustness are in a trade-off regarding features set: EMBER has many features, helping both the detection by containing more information and the attacker who can rely on more fragile features to evade detectors. Finally, our manually selected features required a risk analysis and a good knowledge of attacker capability, so this approach could not be transposed to another domain without the help of an expert. 




\subsection{How do the defense mechanisms compare?}

\begin{table*}[t]
\begin{center}
\begin{tabular}{ l | l | cc }

\textbf{Protection} & \textbf{Model} & \multicolumn{2}{c}{\textbf{EMBER}} 
\\ \midrule

&\textbf{}&\textbf{ROC} &\textbf{Robustness} \\
\midrule

PV feature selection & 
Random Forest & 95.2\% & \textbf{100\%}  \\
&Mono. gradient-boosted trees & 86.7\% & \textbf{100\%} \\
&Gradient-boosted trees & 93.8\% & \textbf{100\%} \\
 \midrule
Adversarial training 
&Random Forest & \textbf{97.6\%}& 94.5\%  \\
&Mono. gradient-boosted trees & 92.7\% & 95.5\%  \\
&Gradient-boosted trees & \textbf{97.6\%} & 96.5\%  \\
\midrule
\acronym&Neural network & 93.0\% & 96.0\% \\
\midrule
\acronym{} \& adversarial training& Neural network & 85.5\% & \textbf{100\%} \\
\end{tabular}
\end{center}
\caption{Comparing the performance of different models with PV feature selection, adversarial training, \acronym, and \acronym~with adversarial training}
\label{tab:acc-rob-protection}
\vspace*{-15pt}
\end{table*}

In this experiment, we compare the several defense mechanisms, namely PV feature selection, adversarial training, \acronym, and \acronym~combined with adversarial training. The results are shown in Table~\ref{tab:acc-rob-protection}. 
As we can see, the PV feature selection approach yields the most robust models but with some detection performance penalties. In fact, this method shows that with access to adversarial examples, the automatic extraction of non-fragile features can yield performances comparable, and even better, to the manual extraction by an expert. 
Adversarial training can be used to get models that are quite robust (about 95\% of attacks failed) without performance penalty compared to the results from Table \ref{tab:acc-rob}. So, compared to the PV feature selection, they are slightly less robust but more effective at discriminating malware from goodware. 
Our framework \acronym{} provides similar results in terms of robustness as adversarial training. 
However, these similar results should not be interpreted as \acronym{} working similarly to adversarial training: as shown by the last line in Table \ref{tab:acc-rob}, both methods can be combined to obtain an even more robust model (in fact, no attack succeeded), at the cost of some performance penalty.

\subsection{How do the feature selection methods compare?}

\begin{table*}[t]
\centering
\begin{tabular}{l | c | c | c}
&  \textbf{PV feature selection} & \textbf{\acronym{} selection} & \textbf{Intersection} \\ 
\midrule
Byte & 0\% & 84.9\% & 0\%\\
Strings & 1.9\% & 94.2\% & 1.9\%\\
General & 30.0\% & 60.0\% & 30.0\% \\
Header & 77.4\% & 83.9\% & 64.5\% \\
Section & 55.2\% & 76.5\% & 40.8\% \\
Imports & 44.5\% & 66.5\% & 22.2\% \\
Exports & 100\% &49.2\% & 49.2\%\\
Data directories & 46.7\% & 90.0\% & 43.3\% \\
\end{tabular}
\caption{Features kept by two features selection techniques and their intersection}
\label{tab:fs}
\vspace*{-15pt}
\end{table*}


EMBER features can be regrouped into several categories. The proportion of features kept for each category is detailed in Table~\ref{tab:fs}.  By looking at the intersection between both feature sets, we can conclude that the features selected by \acronym{} are approximately supersets of the features selected by the PV feature selection. This result shows that \acronym{} retrieve similar features and can exploit more features by linear combinations. This is typically the case for byte and strings: the PV feature selection considers these features to be fragile, but by combining them, \acronym{} is able to create non-fragile features. Remarkably, some features are dropped, notably the Exports features. This is probably because these features are not useful for discriminating goodware from malware. Since \acronym{} learns jointly the linear layer that selects features and the monotonic model that performs the detection, it can drop irrelevant features. 

\subsection{Ablation study of \acronym}
\label{sec:ablation}

\begin{table*}[t]
\begin{center}
\begin{tabular}{ c | c | c | c }

\textbf{Linear layer} & \textbf{Monotonicity} & \textbf{ROC AUC} & \textbf{Robustness}
\\
\midrule
$\times$ & $\times$ & 91.6\% & 82.0\% \\
\checkmark & $\times$ & \textbf{94.3}\% & 91.0\%\\
$\times$ & \checkmark & 87.4\% & 71.5\% \\
\checkmark & \checkmark & 93.0\% & \textbf{96.0\%} \\
\end{tabular}
\end{center}
\caption{ROC AUC and robustness of \acronym{} with some components removed}
\label{tab:ablation}
\vspace*{-15pt}
\end{table*}

The results of the ablation study 
are summarized in Table~\ref{tab:ablation}.
The baseline neural network has an AUC of 91.6\% and a robustness of 82.0\%, as seen in the previous section. With the linear layer, the robustness is increased to 91.0\%. This is consistent with our previous conclusion: the linear layer acts as a feature selection and can ditch fragile features.
The neural network with the monotonicity constraint has a far lower ROC AUC: 87.4\%, and its robustness is also low: 71.5\%. This is consistent with the result of other monotonic models when used on EMBER.
The model with both linear layer and monotonic constraints, as proposed in Figure~\ref{fig:arch}, has a much higher robustness (96.0\%), and its ROC AUC is close to the baseline model, demonstrating how the joined effect of the linear layer and the monotonic constraints enable the detector to be accurate and robust.


\section{Conclusion and future work}
\label{sec:conclusion}

This article focuses on the robustness of machine-learning models against black-box adversarial attacks in the context of malware analysis. In this domain, adversarial attacks are not required to rely on imperceptible perturbations but need to preserve the semantics of malware. Such attacks typically rely on semantics-preserving transformations, like API call addition or n-grams modification. We propose to use manually selected features with a monotonic model to obtain a \textit{robust by design} classifier. We also characterize robust classifiers and deduce a framework named \acronym{} that can use adversarial examples to train a linear layer that selects non-fragile features. This model is \textit{empirically robust by design}, meaning that with enough adversarial examples, it converges to a robust classifier. 
In future work, we will adapt \acronym{} to other security-related data, such as robust classifiers of network packets or robust malware dynamic analysis.


\section*{Acknowledgements}

This work has benefited from a government grant managed by the National Research Agency under France 2030 with reference “ANR-22-PECY-0007,” and has been partially supported by Inria under the SecGen collaboration between Inria, CentraleSupélec and CISPA Helmholtz Center for Information Security.

\bibliographystyle{splncs04}
\bibliography{refs}

\end{document}